\documentclass[12pt]{article}

\usepackage{dsfont}

\usepackage{calrsfs}

\usepackage{amsmath, amsthm, amssymb, amsfonts}
\usepackage[shortlabels]{enumitem}

\usepackage[longnamesfirst, comma]{natbib}
\usepackage{float}

\usepackage[usenames,dvipsnames]{color}

\usepackage{tikz}

\usetikzlibrary{arrows,decorations.pathmorphing,backgrounds,positioning,fit,automata}

\usetikzlibrary{arrows}
\usetikzlibrary{petri}
\usetikzlibrary{topaths}

\usepackage{indentfirst,calc,euscript}
\usepackage{setspace}
\usepackage[reals]{layout}
\usepackage{xr}

\usepackage{sgame}
\usepackage{subfigure}

\newcommand{\diff}{\ \mathrm{d}}

\newcommand{\ignore}[1]{}

\newtheorem{theorem}{Theorem}

\newtheorem{claim}[theorem]{Claim}

\newtheorem{lemma}[theorem]{Lemma}

\newtheorem{proposition}[theorem]{Proposition}

\theoremstyle{definition}

\newtheorem{definition}[theorem]{Definition}

\newtheorem{example}[theorem]{Example}

\newtheorem{remark}[theorem]{Remark}

\numberwithin{equation}{section}

\numberwithin{theorem}{section}

\begin{document}

\title{Existence of equilibria in countable games: \\
an algebraic approach}

\author{Valerio Capraro%
\thanks{Supported by Swiss SNF Sinergia project CRSI22-130435 }\\
   Institut de Math\'ematiques\\
   Universit\'e de Neuch\^atel\\
   Rue Emile-Argand 11 \\
   CH-2000 Neuch\^atel \\
   Switzerland \\
   \texttt{valerio.capraro@unine.ch}
   \and
   Marco Scarsini\\
   Dipartimento di Economia e Finanza\\
   LUISS\\
   Viale Romania 12\\
   I--00197 Roma, Italy\\
   \texttt{marco.scarsini@luiss.it}
   }

\date{\today}

\maketitle

\thispagestyle{empty}

\begin{abstract}

Although mixed extensions of finite games always admit equilibria, this is not the case for countable games, the best-known example being Wald's pick-the-larger-integer game. Several authors have provided conditions for the existence of equilibria in infinite games. These conditions are typically of topological nature and are rarely applicable to countable games. Here we establish an existence result for the equilibrium of countable games when the strategy sets are a countable group and the payoffs are functions of the group operation. In order to obtain the existence of equilibria, finitely additive mixed strategies have to be allowed. This creates a problem of selection of a product measure of  mixed strategies. We propose a family of such selections and prove existence of  an equilibrium that  does not depend on the selection. As a byproduct we show that if finitely additive mixed strategies are allowed, then Wald's game admits an equilibrium. We also prove existence of equilibria for nontrivial extensions of matching-pennies and rock-scissors-paper. Finally we extend the main results to uncountable games.




\end{abstract}

\thispagestyle{empty}

\section{Introduction}

In his celebrated theorem 
\citet{Nas:PNASUSA1950, Nas:AM1951} used fixed point theorems to prove that any finite game admits an equilibrium in mixed strategies. The result fails in general if the strategy sets are not finite. Several authors have provided conditions under which even infinite games admit an equilibrium. Among them \citet{Deb:PNASUSA1952}, who assumed convexity and compactness of the strategy sets and continuity and quasi-concavity of the payoffs, and
\citet{Gli:PAMS1952}, who assumed compactness of the strategy sets and continuity of the payoff functions. The same year \citet{Fan:PNASUSA1952} extended Kakutani's fixed point theorem, this way permitting a generalization of Nash's existence theorem like the one in \citet{Gli:PAMS1952}.

One stream of literature considered existence theorems under various conditions that allow discontinuous payoff functions: see, e.g. 
\citet{DasMas:RES1986,
Sim:RES1987, 
SimZam:E1990,
Ren:E1999, 
Car:IJGT2005,
Car:GEB2010,
BarSoz:Rochestermimeo2010,
BarGovWil:mimeo2012,
BicLar:mimeo2012},
papers in 
\citet{Car:ET2011a}, and references therein.

\citet{Wal:AM1945} considered the case where the strategy set of either one or both players is countable and showed that the mixed extension of a game has a value if one of the strategy sets is finite, but in general it doesn't if they are both countable.

One way to overcome the lack of equilibria in some games is to enlarge the set of mixed strategy by including also finitely additive probability measures.
For the probabilistic and decision-theoretical foundations of the use of finitely additive probability measures, we refer the reader to
\citet{deF:EINAUDI1970, DeF:Wiley1972, deF:Springer2008}, and
\citet{Sav:Dover1972}.
\citet{DubSav:Dover1976} used finitely additive measures extensively in their approach to gambling.

The main issue along this road is that in general a mixed extension is not well defined. Given two measures $\mu_{1}, \mu_{2}$ on the power sets of $S_{1}$ and $S_{2}$, respectively, a product measure $\mu_{1}\otimes\mu_{2}$  is uniquely defined only on the algebra generated by the cylinders and can be extended in a non-unique way to the power set of $S_{1} \times S_{2}$. As a consequence, Fubini's theorem cannot be applied to this situation and in general the order of integration of a double integral matters.

To obviate this drawback several solutions were proposed in the framework of zero-sum two-person games. Both
\citet{Yan:TPA1970} and
\citet{Kin:JOTA1983}
defined the expected value of the payoff to be an arbitrary fixed value, whenever Fubini's theorem cannot be used, and proved this way the existence of a value for the game.
\citet{HeaSud:AMS1972} instead proved the existence of a value by selecting  the product measure that corresponds to a fixed order of integration.
\citet{SheSei:BASE1996} used an approach, whose generalization we follow in our paper, that  selects as product measure a convex combination of the measures obtained by interchanging the order of integration. In order to obtain the existence of a value they need the condition that, in their words, there exist some maxmin $\mu_{1}$-strategy where each good $\mu_{2}$-reply to $\mu_{1}$ is close to one of some finite collections of $\mu_{2}$.

Finitely additive mixed strategies have been used  by
\citet{MaiSud:IJGT1993, MaiSud:IJGT1998}  in the framework of zero-sum stochastic games.
\citet{Cot:JET1991} considered finitely additive strategies in correlated equilibria; \citet{Sti:JET2011a} showed the limitation of his approach and proposed an alternative one.

\citet{Mar:IJGT1997} proved the existence of Nash equilibria in finitely additive mixed strategies under purely measure-theoretic conditions and connected this to the existence of $\varepsilon$-equilibria with countably additive mixed strategies. His analysis needs to restrict attention to payoff functions that are measurable with respect to the algebra generated by the cylinders. 

\citet{HarStiZam:GEB2005} dealt with a class of games called nearly compact and continuous  and proved existence of equilibria for any game in this class via a continuous compact imbedding in a larger game. They showed the use and limitations of finitely additive mixed strategies for these games.
 
\citet{Sti:GEB2005} devoted his attention to games that are not nearly compact and continuous and considered several classes of equilibria for these games showing advantages of disadvantages for each of them. The use of finitely additive strategies is fundamental in his analysis.

\citet{MyeRen:mimeo2012} have recently proposed a new notion of equilibrium for infinite games using finitely additive mixed strategies that arise from suitable finite approximations. 

\citet{CapMor:IJGT2012} proved existence of equilibria for a family of zero-sum two-person games on semigroups when feasible mixed strategies are restricted to a suitable subclass of the class of  finitely additive probability measures.

In this paper we prove an existence result for Nash equilibria of countable games  by imposing some algebraic conditions on the payoff functions. The strategy set of each player is assumed to be a countable group and the payoff functions depend on their arguments only through the group operation.  No topological condition is required. We allow finitely additive mixed strategies defined on the power set of the group. As mentioned before, this requires some care since the product of finitely additive measures is not uniquely defined on the power set of the Cartesian product of the groups, but only on the algebra generated by the cylinders. Since we want to integrate payoff functions that are not measurable with respect to this algebra, we need to select a suitable extension of the product measure. We propose a natural class of extensions by considering an average over all possible orders of integration. We show that the equilibrium exists and it does not depend on the way we choose this average. We characterize the equilibrium strategies and prove that they are the invariant means over the group that solve a suitable variational problem. The equilibrium payoffs have a very simple form. To avoid drowning our result in a sea of measure-theoretic technicalities, we first  develop the theory for games on countable groups; then we show how it can be extended to uncountable groups, under suitable assumptions. 

The paper is organized as follows.
Section~\ref{se:Main} describes the model and states the main result.
Section~\ref{se:Variational} proves a variational principle that is of interest \emph{per se} and is used in the proof of the main result.
Section~\ref{se:Proof} contains the proof of the main result.
Section~\ref{se:Generalization} examines some interesting generalizations.
Section~\ref{se:Examples} considers several examples.
Section~\ref{se:Uncountable} deals with uncountable groups.
Section~\ref{se:Conclusions} provides some conclusive remarks.

\section{Group games}\label{se:Main}

\subsection{Finite games}

Consider the classical matching pennies game
\begin{equation}\label{eq:matchingpennies}
\begin{array}[c]{c|rr|rr|}
\multicolumn{1}{c}{} & \multicolumn{2}{c}{A} &
\multicolumn{2}{c}{B} \\
\cline{2-5}
A & -1,& 1 & 1,& -1  \\
\cline{2-5}
B & 1,& -1 & -1,& 1  \\
\cline{2-5}
\end{array}
\end{equation}
We know that the unique equilibrium of this game is the profile of mixed strategies $((1/2,1/2), (1/2,1/2))$. 

Notice that the matching pennies game can be re-written as follows. Make the set $\{A,B\}$  a finite group\footnote{A set $G$ with a binary operation $*$ is called a \emph{group} if the operation is associative, it has a unit element, and every element has an inverse. If the operation is commutative, then the group is called \emph{abelian}.}
by endowing it with the binary operation $*$ defined as
\[
A*A=B*B= B, \quad A*B=B*A = A.
\]
Define $\phi : \{A,B\} \to \mathbb{R}$ as follows:
\[
\phi(x) = 
\begin{cases}
1 & \text{for $x=A$}, \\
-1 & \text{for $x=B$}.
\end{cases}
\]
Consider a game played by players $1$ and $2$, where each player's pure strategy set is $\{A,B\}$ and the payoffs are 
\[
u_{1}(x,y) = - u_{2}(x,y) = \phi(x*y), \quad \text{for $x,y \in \{A,B\}$}.
\]
The game that we just described is nothing else than the matching pennies game defined in \eqref{eq:matchingpennies}.

This suggests the following generalization. 
Consider a finite group $(G, *)$ and $N$ functions $\phi_{1}, \dots, \phi_{N} : G \to \mathbb{R}$. 
Given a set of players $P=\{1, \dots, N\}$, for $i\in P$ let $u_{i} : G^{N} \to \mathbb{R}$ be defined as 
\begin{equation}\label{eq:uphifinite}
u_{i}(x_{1}, \dots, x_{N}) = \phi_{i}(x_{1} * \dots * x_{N}).
\end{equation}
For $\boldsymbol{\phi} := (\phi_{1}, \dots, \phi_{N})$, call $\mathcal{G}(P,G, \boldsymbol{\phi})$ the game where the set of players is $P$, each player's set of pure strategies is $G$, and player $i$'s payoff function is given by \eqref{eq:uphifinite}. 
Call $\mathcal{P}(G)$ the set of all probability measures on $2^{G}$. A probability measure  $\lambda \in \mathcal{P}(G)$ is \emph{invariant} if for all  $x, y \in G$ we have $\lambda(x) = \lambda(x*y)$. 
Observe that finite groups have a unique invariant
measure, that is, the uniform measure. We will see in the next sections that a countable
group may have many invariant measures.
\begin{proposition}\label{pr:mainfinite}
The game  $\mathcal{G}(P,G, \boldsymbol{\phi})$ admits an equilibrium in mixed strategies $(\lambda, \dots, \lambda)$, with $\lambda$ invariant on $G$.
\end{proposition}

\begin{proof}
For $\mu_{1}, \dots, \mu_{N} \in \mathcal{P}(G)$, define
\begin{equation*}
u_{i}(\mu_{1}, \dots, \mu_{N}) = \sum_{x_{1} \in G} \dots \sum_{x_{N} \in G} u_{i}(x_{1}, \dots, x_{N}) \mu_{1}(x_{1}) \cdots \mu_{N}(x_{N}).
\end{equation*}
Then we have to prove that for all $i \in P$ and all $\mu_{i} \in \mathcal{P}(G)$ we have
\begin{equation}\label{eq:lambdamufinite}
u_{i}(\lambda, \dots, \lambda) \ge u_{i}(\lambda, \dots, \lambda, \mu_{i}, \lambda, \dots, \lambda). 
\end{equation}
Notice that, by definition of $u_{i}$, for all $j \in P$, for all $x_{1}, \dots, x_{j-1}, x_{j+1}, \dots, x_{N} \in G$ we have
\begin{align*}
\sum_{x_{j} \in G} u_{i}(x_{1}, \dots, x_{j}, \dots, x_{N}) \lambda(x_{j}) &= \sum_{x_{j} \in G} \phi_{i}(x_{1} * \dots * x_{j} * \dots * x_{N}) \lambda(x_{j}) \\
&=  \sum_{y_{j} \in G} \phi_{i}(y_{j}) \lambda(x_{j-1}^{-1} * \dots * x_{1}^{-1} * y_{j} * x_{N}^{-1} * \dots * x_{j+1}^{-1}) \\
&=  \sum_{y_{j} \in G} \phi_{i}(y_{j}) \lambda(y_{j}),
\end{align*}
where we used the change of variable $y_{j} = x_{1} * \dots * x_{j} * \dots * x_{N}$. Hence \begin{align*}
u_{i}(\lambda, \dots, \lambda) &= \sum_{x_{1} \in G} \dots \sum_{x_{N} \in G} \phi_{i}(x_{1} * \dots * x_{N}) \lambda(x_{1}) \cdots \lambda(x_{N}) \\
&=  \sum_{y_{j} \in G} \phi_{i}(y_{j}) \lambda(y_{j}) \\
&= \sum_{x_{1} \in G} \dots \sum_{x_{i} \in G} \dots \sum_{x_{N} \in G} \phi_{i}(x_{1} * \dots * x_{N}) \lambda(x_{1}) \cdots \mu_{i}(x_{i}) \cdots \lambda(x_{N}) \\
&= u_{i}(\lambda, \dots, \lambda, \mu_{i}, \lambda, \dots, \lambda),
\end{align*}
that is, \eqref{eq:lambdamufinite} holds.
\end{proof}
\citet{Mor:MM2010} proved an analogous result for zero-sum games.

In the rest of the paper we will find conditions for the existence of equilibria in countable games, that, among other things, allow to extend Proposition~\ref{pr:mainfinite} to the case of countable strategy sets.

\subsection{Countable games}

Given a set of players $P=\{1, \dots, N\}$, a countable set $S$ and bounded functions $u_i: S^{N} \to [0,1]$, $i \in P$, consider a game $\mathcal{G}= \langle P, S,(u_{i})_{i\in P}\rangle$, where $S$ is the strategy set of all players, and $u_i$ is the payoff function of player $i$. 

As mentioned in the Introduction, existence of mixed equilibria may fail if only countably additive mixed strategies are allowed. Therefore we consider a mixed extension of the game $\mathcal{G}$ where the space of mixed strategies is $\mathcal{P}(S)$, the space of all finitely additive probability measures on $S$. When doing this, a selection problem immediately arises. Given $\mu_{1}, \dots, \mu_{N} \in \mathcal{P}(S)$, a product measure $\otimes_{i=1}^{N} \mu_{i}$ is uniquely defined only on the algebra generated by the cylinders $S \times \dots \times S \times A \times S \dots  \times S$, for all $A \subset S$. This product measure can be (non-uniquely) extended to the power set $2^{S \times \dots \times  S}$. Different extensions correspond to different values of the expected payoff $\int_{S \times \dots \times S} u \diff \otimes_{i=1}^{N} \mu_{i}$. Here we consider a parametric class of possible extensions that has the advantage of being easily computable. Its simpler bivariate version has been used for zero-sum two-person games  by \citet{SheSei:BASE1996}.
Call $\Sigma(P)$ the space of permutations of $P$. Let  $\nu \in \mathcal{P}(\Sigma(P))$ and $\mu_{1}, \dots, \mu_{N} \in \mathcal{P}(S)$.
For $i \in P$, let  $u_i : S^{N} \to [0,1]$. Define
\begin{equation}\label{eq:ualpha}
u_i^{\nu}(\mu_{1}, \dots, \mu_{N}) :=
\sum_{\pi \in \Sigma(P)} \nu(\pi) \int_{S} \dots \int_{S} u_{i}(x_{1}, \dots, x_{N}) \diff \mu_{\pi(1)}(x_{\pi(1)}) \dots \diff \mu_{\pi(N)}(x_{\pi(N)}).
\end{equation}
This clearly defines an extension $\mu_{1} \boxtimes_{\nu} \dots \boxtimes_{\nu} \mu_{N}$ of $\otimes_{i=1}^{N} \mu_{i}$ to $2^{S \times \dots \times S}$ as follows.
 For $A \subset S \times \dots \times S$
\begin{equation}\label{eq:otimesalpha}
\mu_{1} \boxtimes_{\nu} \dots \boxtimes_{\nu} \mu_{N}(A) = \sum_{\pi \in \Sigma(P)} \nu(\pi) \int_{S} \dots \int_{S} \mathds{1}_{A}(x_{1}, \dots, x_{N}) \diff \mu_{\pi(1)}(x_{\pi(1)}) \dots \diff \mu_{\pi(N)}(x_{\pi(N)}),
\end{equation}
where $\mathds{1}_{A}$ is the indicator function of the set $A$.

For properties of integration with respect to finitely additive measures we refer the reader to \citet{Hil:TAMS1934},
\citet{DunSch1:Wiley1988}, \citet{DeF:Wiley1972}, and \citet{BhaBha:AcademicPress1983}. Since every bounded function on a countable set is integrable with respect to any finitely additive probability measure, we do not need any measure-theoretical assumption.

We can now state our main theorem. 
Let $(G, *)$ be a countable group%
\footnote{In the whole paper all countable groups will be endowed with the discrete topology.} 
and given $\phi_{1}, \dots, \phi_{N} : G \to [0,1]$, define
\begin{equation}\label{eq:uphi}
u_i(x_{1}, \dots, x_{N})=\phi_{i}(x_{1}* \dots * x_{N}).
\end{equation}
For $\boldsymbol{\phi} = (\phi_{1}, \dots, \phi_{N})$ call $\mathcal{G}(P,S, \boldsymbol{\phi}, \nu)$ the mixed extension of the game $\mathcal{G}$ when $u_{i}$ is defined as in \eqref{eq:uphi} and the product measure of the finitely additive mixed strategies is selected as in \eqref{eq:otimesalpha}. 

\begin{theorem}\label{th:main}
If $(G,*)$ is a countable abelian group, then the game $\mathcal{G}(P,G, \boldsymbol{\phi}, \nu)$ admits a Nash equilibrium that does not depend on $\nu$.
\end{theorem}

A more general version of Theorem~\ref{th:main} will be proved in Section~\ref{se:Proof}.

\section{A variational principle for FC-groups}\label{se:Variational}

In this section we prove a preliminary result that has some interest \emph{per se} since it represents a new variational principle for a useful class of groups that we now define.
\begin{definition}\label{def:fcgroups}
A countable group $G$ is called an \emph{FC-group} if for all $g\in G$, the conjugacy class $\{h*g*h^{-1} : h\in G\}$ is finite.
\end{definition}

FC-groups have been introduced by
\citet{Bae:DMJ1948} and
\citet{Neu:PLMS1951}. Among others, abelian groups are FC, since every conjugacy class is a singleton.

Let $G$ be a countable group, $A\subset G$ and $g\in G$. Fix the following notation
\[
g*A=\{g*a : a\in A\}\qquad A*g=\{a*g : a\in A\}.
\]

\begin{definition}\label{def:invariantmeasure}
A finitely additive probability measure $\mu$ on the power set of $G$ is called
\begin{itemize}
\item  \emph{left-invariant mean}, if $\mu(A)=\mu(g*A)$, for all $g\in G$ and $A\subset G$,
\item  \emph{right-invariant mean}, if $\mu(A)=\mu(A*g)$, for all $g\in G$ and $A\subset G$,
\item  \emph{invariant mean}, if it is both left- and right-invariant.
\end{itemize}
For a given countable group $G$, we call  $\mathcal{L}(G)$,  $\mathcal{R}(G)$ and $\mathcal{I}(G)$ the class of all left-invariant, right-invariant and invariant means on $G$, respectively.
\end{definition}

It is well-known that the existence of a left-invariant mean is equivalent to the existence of a right-invariant mean, that is equivalent to the existence of an invariant mean\footnote{Indeed, given a left-invariant mean $\lambda$, one can define a right-invariant mean $\rho$ by setting $\rho(A)=\lambda(A^{-1})$, where $A^{-1} = \{a^{-1} : a \in A\}$. Now one can define an invariant mean $\mu$ by the formula $\mu(A)=\int_G\lambda(A * g^{-1})\diff \rho(g)$.}.

\begin{definition}\label{def:amenable}
A countable group is called \emph{amenable} if it admits a left-invariant mean.
\end{definition}
Amenable groups have been introduced by
\citet{vNe:FM1929} in relation to the Tarski paradox and they form a hugely studied class of groups still nowadays. Every finite group is amenable, just taking the uniform measure; abelian groups are amenable by a standard but non-trivial argument making use of the Markov-Kakutani fixed point theorem. The simplest example of a non-amenable group is the free group on two generators\footnote{The free group on two generators, say $x$ and $y$, is the group of all words in the letters $x,x^{-1},y,y^{-1}$, equipped with the operation of concatenation of words, where only the simplifications $x*x^{-1}=x^{-1}*x=y*y^{-1}=y^{-1}*y=e$ are allowed, being $e$ the empty word. It was observed by von Neumann himself that this group, usually denoted by $\mathbb F_2$, is not amenable. A celebrated example of Ol{$'$}{\v{s}}anski{\u\i} shows the existence of non-amenable groups which do not contain $\mathbb F_2$ \citep[see][]{Ols:UMN1980}.}.

We use the following theorem, which appeared in \citet[Theorem~3.2]{Pat:PJM1979}.
\begin{theorem}\label{th:Paterson}
Let $G$ be a countable amenable FC-group. Then
\[
\mathcal{R}(G)=\mathcal{L}(G)=\mathcal{I}(G).
\]
\end{theorem}
This means that we have no distinctions between left- and -right-invariant means.

Given a countable amenable group $G$,  $\ell^\infty(G)$ denotes the Banach space of all bounded real-valued function on $G$. The main result of this section is the following variational principle.
\begin{theorem}\label{th:variationalprinciple}
If $G$ is a countable amenable FC-group and  $f:G\rightarrow[0,1]$, then for all $\pi \in \Sigma(P)$ and all $\lambda_{2}, \dots, \lambda_{N} \in\mathcal{I}(G)$ the functional $\Psi : \mathcal{P}(G) \to \mathbb{R}$ defined as
\[
\Psi(\mu) = \int \dots \int \int f(x_{1}* \dots * x_{N}) \diff \mu(x_{\pi(1)}) \diff \lambda_{2}(x_{\pi(2)}) \dots \diff \lambda_{N}(x_{\pi(N)})
\]
attains its maximum at some $\lambda \in\mathcal{I}(G)$.
\end{theorem}

\begin{remark}
The space $\mathcal{P}(G)$ is a closed subset of the unit ball of the dual of $\ell^{\infty}(G)$ and therefore is compact in the weak* topology by the Banach-Alaoglu theorem \citep[see, e.g.,][Theorem 5.93]{AliBor:Springer2006}
Nevertheless the existence of the maximum for the functional $\Psi$ is not automatic, since it is not continuous.
To see this take $G=(\mathbb{Z}, +)$, $N=2$,  and consider the functional
\[
\Psi(\mu) := \int \int \mathds{1}_{\mathbb{N}}(x+y) \diff \mu(x) \diff \lambda(y),
\]
with $\lambda \in \mathcal{I}(\mathbb{Z})$ such that $\lambda(\mathbb{N})=1$. Call $\mu_{\alpha}$ a net of probability measures having finite support in $-\mathbb{N}$ and converging to some $\mu \in \mathcal{I}(\mathbb{Z})$. Observe that $\mu(\mathbb{N}) = 0$. If $\Psi$ were continuous we would have
\begin{equation}\label{eq:continuousPsi}
\Psi(\mu) = \lim_{\alpha} \Psi(\mu_{\alpha}) .
\end{equation}
But 
\begin{align*}
\Psi(\mu) &= \int \int \mathds{1}_{\mathbb{N}}(x+y) \diff \mu(x) \diff \lambda(y) = 0, \\
\Psi(\mu_{\alpha}) &= \int \int \mathds{1}_{\mathbb{N}}(x+y) \diff \mu_{\alpha}(x) \diff \lambda(y) \\
 &= \int \int \mathds{1}_{\mathbb{N}}(x+y) \diff \lambda(y) \diff \mu_{\alpha}(x) = 1 \quad \text{for all } \alpha,
\end{align*}
which contradicts \eqref{eq:continuousPsi}.
\end{remark}

Denote
\begin{equation*}
I(f)=\left\{\int f(x) \diff \lambda(x) : \lambda\in\mathcal{I}(G)\right\}.
\end{equation*}

The following lemma is folklore and follows from the fact that the set  $\mathcal I(G)$ is convex and weak*-compact, when seen as a subset of the dual of $\ell^\infty(G)$.

\begin{lemma}\label{lem:convex}
The set  $I(f) \subset \mathbb{R}$ is convex and compact.
\end{lemma}

This lemma guarantees that the following number is well defined.
\begin{equation}\label{eq:maxima}
I(f)^{+} := \max I(f).
\end{equation}

\begin{lemma}\label{lem:minimalgain}
Let $G$ be a countable amenable FC-group and $f:G \to [0,1]$. If there exist $\mu\in\mathcal{P}(G)$ and $L\in\mathbb{R}$ such that either 
\begin{equation}\label{eq:intmuL}
\int f(x*y) \diff \mu(x) \geq L \quad \text{for all $y \in G$},
\end{equation}
or
\begin{equation}\label{eq:intmuL2}
\int f(y*x) \diff \mu(x) \geq L \quad \text{for all $y \in G$},
\end{equation}
then there exists $\lambda \in \mathcal{I}(G)$ such that $\int f(x) \diff \lambda(x)\geq L$.
\end{lemma}

Given a set $A$, its cardinality is denoted by $|A|$.

\begin{definition}
A sequence $F_{n}$ of finite subsets of $G$ is called a \emph{left-F{\o}lner sequence} for $G$ if for all $g\in G$ one has
\[
\lim_{n\rightarrow\infty}\frac{|(g*F_{n})\triangle F_{n}|}{|F_{n}|}=0
\]
and a \emph{right-F{\o}lner sequence} for $G$ if for all $g\in G$ one has
\[
\lim_{n\rightarrow\infty}\frac{|(F_{n}*g)\triangle F_{n}|}{|F_{n}|}=0,
\]
where $\triangle$ stands for the symmetric difference of sets; i.e. $A\triangle B=(A\cup B)\setminus(A\cap B) = (A \setminus B) \cup (B \setminus A)$.
\end{definition}
\citet{Fol:MS1955} proved that such sequences exist for all countable amenable groups.

\begin{proof}[Proof of Lemma~\ref{lem:minimalgain}]
Let $F_{n}$ be a left-F{\o}lner sequence for $G$.
Consider the sequence of measures $\mu_{n}$ defined by
\[
\mu_{n}(A)=\frac{1}{|F_{n}|}\sum_{g\in F_{n}}\mu(A*g)
\]
and let $\lambda$ be a weak* limit of (a subnet $\mu_{c(\alpha)}$ of) this sequence.
First we prove that  $\lambda\in\mathcal{I}(G)$. Indeed, for all $A\subset G$ and for all $h\in G$, one has
\begin{align*}
|\lambda(A*h)-\lambda(A)| &= \lim_\alpha|\mu_{c(\alpha)}(A*h)-\mu_{c(\alpha)}(A)| \\
&=
\lim_\alpha\left|\sum_{g\in F_{c(\alpha)}}\frac{1}{|F_{c(\alpha)}|}\left(\mu_{c(\alpha)}(A*h*g)-\mu_{c(\alpha)}(A*g)\right)\right|.
\end{align*}
Observe that the terms that are not in $(h*F_{c(\alpha)}) \triangle F_{c(\alpha)}$ cancel out. Majorizing with $1$ each of the remaining terms, we get
\[
|\lambda(A*h)-\lambda(A)|\leq \lim_{\alpha}\frac{|(h*F_{c(\alpha)})\triangle F_{c(\alpha)}|}{|F_{c(\alpha)}|}=0,
\]
which proves that $\lambda\in\mathcal{R}(G)$. Theorem~\ref{th:Paterson} implies that  $\lambda\in\mathcal{I}(G)$.
Now we prove that if \eqref{eq:intmuL} holds, then 
$\int f(x) \diff \lambda(x)\geq L$. Indeed, we have
\begin{align*}
\int f(x) \diff \lambda(x) &= \lim_\alpha\int f(x) \diff \mu_{c(\alpha)}(x) \\
&=
\lim_\alpha\int\frac{1}{|F_{c(\alpha)}|}\sum_{g\in F_{c(\alpha)}}f(x) \diff \mu(x*g) \\
&=
\lim_\alpha\int\frac{1}{|F_{c(\alpha)}|}\sum_{g\in F_{c(\alpha)}}f(x*g^{-1}) \diff \mu(x)\\
&=\lim_\alpha\frac{1}{|F_{c(\alpha)}|}\sum_{g\in F_{c(\alpha)}}\int f(x*g^{-1}) \diff \mu(x)\\
&\geq L.
\end{align*}
where the inequality stems from the hypothesis that each of the $|F_{c(\alpha)}| $ summands is larger or equal $L$.
The proof for the case \eqref{eq:intmuL2}  is similar.
\end{proof}

Given a group $G$ we call $G^{N}$ the direct product of $G$, $N$ times, endowed with the component-wise operation, still denoted by $*$, with a little abuse of notation. If $G$ is amenable, then $G^{N}$ is amenable, too \citep{Day:IJM1957}. Furthermore, a simple computation shows that if $G$ is an FC-group, then also $G^{N}$ is an FC-group.

\begin{lemma}\label{lem:product}
Under the hypotheses of Theorem~\ref{th:Paterson} 
\begin{align}
\max\left\{\int f(x_{1}* \dots * x_{N})\diff \sigma(x_{1}, \dots, x_{N}) : \sigma\in\mathcal{I}(G^{N})\right\}=I(f)^{+},
\end{align}
where $I(f)^{+}$ is defined as in \eqref{eq:maxima}.
\end{lemma}

\begin{proof}
For 
\begin{equation}\label{eq:psif}
\psi(x_{1}, \dots, x_{N}) = f(x_{1}* \dots * x_{N})
\end{equation}
define the set
\begin{equation*}
\mathbb{R} \supset \Lambda(\psi)=\left\{\int \psi(x_{1}, \dots, x_{N}) \diff \sigma(x_{1}, \dots, x_{N}) : \sigma\in\mathcal{I}(G^{N})\right\}
\end{equation*}
and call $L=\max \Lambda(\psi)$, which exists by Lemma~\ref{lem:convex}. We have to prove that $L=I(f)^{+}$.\\

\noindent
\textbf{Proof of the inequality $L\geq I(f)^{+}$.}
By Lemma~\ref{lem:convex} there exists $\lambda\in\mathcal{I}(G)$ such that $\int f(x) \diff \lambda(x)=I(f)^{+}$. Let $\lambda^{\otimes N}$ denote the measure on $G^{N}$ defined by the functional
\[
\ell^\infty(G^{N}) \ni \gamma \mapsto \int \dots \int \gamma(x_{1}, \dots, x_{N})\diff \lambda(x_{1}) \dots \diff \lambda(x_{N}).
\]
First we show that $\lambda^{\otimes N} \in\mathcal{I}(G^{N})$. Indeed for all $\gamma \in \ell^{\infty} (G^{N})$ and for all $(g_{1}, \dots, g_{N}) \in G^{N}$ we have
\begin{align*}
&\int \gamma((g_{1}, \dots, g_{N}) * (x_{1}, \dots, x_{N})) \diff \lambda^{\otimes N}(x_{1}, \dots, x_{N}) \\
&\qquad= \int \gamma(g_{1}*x_{1}, \dots, g_{N}* x_{N}) \diff \lambda^{\otimes N}(x_{1}, \dots, x_{N}) \\
&\qquad=\int \dots \int \gamma(g_{1}*x_{1}, \dots, g_{N}* x_{N}) \diff \lambda(x_{1}) \dots \diff \lambda(x_{N})  \\
&\qquad=\int \dots \int \gamma(x_{1}, g_{2} * x_{2, },\dots, g_{N}* x_{N}) \diff \lambda(x_{1}) \dots \diff \lambda(x_{N})  \\
&\qquad=\int \dots \int \gamma(x_{1},  x_{2}, g_{3}* x_{3},\dots, g_{N}* x_{N}) \diff \lambda(x_{1}) \dots \diff \lambda(x_{N})  \\
&\qquad \qquad \vdots \\
&\qquad=\int \dots \int \gamma(x_{1},\dots, x_{N}) \diff \lambda(x_{1}) \dots \diff \lambda(x_{N}),
\end{align*}
where the third equality stems from the fact that $\lambda$ is a left invariant mean, therefore 
\[
\int \gamma(g_{1}*x_{1}, \dots, g_{N}* x_{N}) \diff \lambda(x_{1}) = \int \gamma(x_{1}, g_{2} * x_{2, }\dots, g_{N}* x_{N}) \diff \lambda(x_{1}).
\]
For the forth equality define $\zeta(x_{2}, \dots, x_{N}) = \int \gamma(x_{1}, x_{2}, \dots, x_{N}) \diff \lambda(x_{1})$. Then, since $\lambda$ is a left invariant mean, we have
\begin{equation*}
 \int \zeta(g_{2}* x_{2}, g_{3}* x_{3}, \dots, g_{N} * x_{N}) \diff \lambda(x_{2})  
=  \int\zeta(x_{2}, g_{3}* x_{3}, \dots, g_{N} * x_{N}) \diff \lambda(x_{2}),
\end{equation*}
i.e.,
\begin{multline*}
\int \int \gamma(x_{1, }g_{2}* x_{2}, g_{3}* x_{3}, \dots, g_{N} * x_{N}) \diff \lambda(x_{1}) \diff \lambda(x_{2})  \\
= \int \int \gamma(x_{1}, x_{2}, g_{3}* x_{3}, \dots, g_{N} * x_{N}) \diff \lambda(x_{1}) \diff \lambda(x_{2}).
\end{multline*}
The remaining equalities are obtained analogously. We have then shown that $\lambda^{\otimes N} \in \mathcal{L}(G^{N})$. 
By Theorem~\ref{th:Paterson} applied to the FC-group $G^{N}$ it follows that  $\lambda^{\otimes N} \in \mathcal{I}(G^{N})$. 

Now we show that  $\int f(x_{1} * \dots * x_{N}) \diff \lambda^{\otimes N}(x_{1}, \dots, x_{N}) = I(f)^{+}$. We have
\begin{align*}
\int f(x_{1} * \dots * x_{N}) \diff \lambda^{\otimes N}(x_{1}, \dots, x_{N})  &= \int \dots \int f(x_{1} * \dots * x_{N}) \diff \lambda(x_{1}) \dots \diff \lambda(x_{N})  \\
&= \int \dots \int f(x_{1}) \diff \lambda(x_{1}) \dots \diff \lambda(x_{N}) \\
&= \int \dots \int I(f)^{+} \diff \lambda(x_{2}) \dots \diff \lambda(x_{N}) \\
&= I(f)^{+}.
\end{align*}
This shows that $I(f)^{+}\in \Lambda(\psi)$ and therefore $I(f)^{+}\leq L$.\\

\noindent
\textbf{Proof of the inequality $L\leq I(f)^{+}$.}
Let $\overline{\sigma} \in \mathcal{I}(G^{N})$ be such that 
\[
\int f(x_{1} * \dots * x_{N}) \diff \overline{\sigma}(x_{1},\dots,x_{N})=L
\] 
(such a measure exists by Lemma~\ref{lem:convex} applied to the group $G^{N}$ and to the function $\psi(x_{1}, \dots, x_{N})=f(x_{1} * \dots * x_{N})$) and let $\sigma_{\alpha}$ be a net of countably additive probability measures on $G^{N}$ converging to $\overline{\sigma}$ in the weak* topology.   Define a countably additive probability measure $\mu_{\alpha}$ on $G$ by setting for all $x_{1} \in G$
\[
\mu_{\alpha}(x_{1})=\sum_{(x_{2}, \dots, x_{N}) \in G^{N-1}}\sigma_{\alpha}(x_{1}* (x_{2} * \dots * x_{N})^{-1}, x_{2}, \dots, x_{N}).
\]
To show that this is indeed a countably additive probability measure notice that for each $x_{1} \in G$ we have $\mu_{\alpha}(x_{1})\geq 0$ and therefore it suffices to show that $\sum_{x_{1}\in G}\mu_{\alpha}(x_{1})=1$. This follows from the observation that $\mu_{\alpha}(x_{1}) = \sigma_{\alpha}(A_{x_{1}})$, where 
\[
A_{x_{1}}=\left\{(x_{1}* (x_{2} * \dots * x_{N})^{-1}, x_{2}, \dots, x_{N}) : (x_{2}, \dots, x_{N}) \in G^{N-1}\right\}
\]
and the fibers $A_{x_{1}}$ form a partition of $G^{N}$. 

Now, let $\mu$ be any weak* limit of a subnet, denoted by $\mu_\beta$, of the net $\mu_{\alpha}$. For any $g\in G$, one has
\begin{align*}
\int_G f(g*x_{1}) \diff \mu(x_{1}) &= \lim_{\beta}\int_G f(g*x_{1}) \diff \mu_{\beta}(x_{1})\\
&=\lim_{\beta}\sum_{x_{1}\in G}f(g*x_{1}) \mu_{\beta}(x_{1})\\
&=\lim_{\beta}\sum_{x_{1}\in G} \sum_{(x_{2}, \dots, x_{N}) \in G^{N-1}}f(g*x_{1}) \sigma_{\beta}(x_{1}*(x_{2} * \dots * x_{N})^{-1}, x_{2}, \dots, x_{N})\\
&=\lim_{\beta} \sum_{(x_{2}, \dots, x_{N}) \in G^{N-1}} \sum_{x_{1}\in G} f(g*x_{1}) \sigma_{\beta}(x_{1}*(x_{2} * \dots * x_{N})^{-1}, x_{2}, \dots, x_{N})\\
&= \lim_{\beta}\sum_{(x_{2}, \dots, x_{N}) \in G^{N-1}} \sum_{z\in G*(x_{2} * \dots * x_{N})^{-1}}f(g*z*x_{2} * \dots * x_{N})
\sigma_{\beta}(z, x_{2}, \dots, x_{N}),
\end{align*}
where in the last equality we put $z=x_{1}*(x_{2} * \dots * x_{N})^{-1}$.
Observe that in the fourth equality we can exchange the order of summation since the series are nonnegative and convergent. For the same reason we can now replace $G*(x_{2} * \dots * x_{N})^{-1}$ with $G$ (the mapping $x\mapsto x*(x_{2} * \dots * x_{N})^{-1}$ is a permutation). Therefore, using again \eqref{eq:psif}, we have
\begin{align*}
&\lim_{\beta}\sum_{(x_{2}, \dots, x_{N}) \in G^{N-1}} \sum_{z\in G*(x_{2} * \dots * x_{N})^{-1}}f(g*z*x_{2} * \dots * x_{N})
\sigma_{\beta}(z, x_{2}, \dots, x_{N}) \\
&\qquad=
\lim_{\beta}\sum_{(z, x_{2}, \dots, x_{N}) \in G^{N}} f(g*z*x_{2} * \dots * x_{N})
\sigma_{\beta}(z, x_{2}, \dots, x_{N}) \\
&\qquad=
\lim_{\beta}\int  f(g*z*x_{2} * \dots * x_{N}) \diff
\sigma_{\beta}(z, x_{2}, \dots, x_{N}) \\
&\qquad=
\lim_{\beta}\int  \psi((g, 1_{G}, \dots, 1_{G}) * (z, x_{2}, \dots, x_{N})) \diff
\sigma_{\beta}(z, x_{2}, \dots, x_{N}) \\
&\qquad=
\int  \psi((g, 1_{G}, \dots, 1_{G}) * (z, x_{2}, \dots, x_{N})) \diff
\overline{\sigma}(z, x_{2}, \dots, x_{N}) \\
&\qquad=
\int  \psi(z, x_{2}, \dots, x_{N}) \diff
\overline{\sigma}(z, x_{2}, \dots, x_{N}) \\
&\qquad=
\int  f(z * x_{2} * \dots * x_{N}) \diff
\overline{\sigma}(z, x_{2}, \dots, x_{N}) \\
&\qquad=L.
\end{align*}
We have proved that $\int_G f(g*x_{1}) \diff \mu(x_{1}) = L$ for all $g\in G$.
Therefore, by Lemma \ref{lem:minimalgain},  there exists $\lambda\in\mathcal{I}(G)$ such that $\int f(x) \diff \lambda(x)\geq L$. It follows that $I(f)^{+}\geq L$.
\end{proof}

\begin{proof}[Proof of Theorem~\ref{th:variationalprinciple}]
Call $\boldsymbol{\lambda} = (\lambda_{2}, \dots, \lambda_{N})$ and
\begin{align}
S_{\boldsymbol{\lambda}, \pi}=\sup_{\mu\in\mathcal{P}(G)}\left\{
\int \dots \int \int f(x_{1}* \dots * x_{N}) \diff \mu(x_{\pi(1)}) \diff \lambda_{2}(x_{\pi(2)}) \dots \diff \lambda_{N}(x_{\pi(N)})
\right\}.
\end{align}
By Lemma \ref{lem:convex}, we know that $S_{\boldsymbol{\lambda}, \pi} \geq I(f)^{+}$, for all $\boldsymbol{\lambda}$ and $\pi$. 
We recall that the value $I(f)^{+}$ is attained, basically by definition, by an invariant mean. So it suffices to show that $S_{\boldsymbol{\lambda}, \pi} = I(f)^{+}$.
Now, by contradiction, suppose that there exist $\mu\in\mathcal P(G)$ and $\lambda_{2}, \dots, \lambda_{N}\in\mathcal{I}(G)$ such that
\[
\int \dots \int \int f(x_{1}* \dots * x_{N}) \diff \mu(x_{\pi(1)}) \diff \lambda_{2}(x_{\pi(2)}) \dots \diff \lambda_{N}(x_{\pi(N)}) =: L>I(f)^{+}.
\]
Call $\sigma$ the measure on $G^{N}$ defined by the functional
\begin{multline*}
\ell^\infty(G^{N}) \ni \gamma \mapsto \int_{G^{N}}\gamma(x_{1}, \dots, x_{N}) \diff \sigma(x_{1}, \dots, x_{N}) \\
=\int \dots \int \int \gamma(x_{1}, \dots,  x_{N}) \diff \mu(x_{\pi(1)}) \diff \lambda_{2}(x_{\pi(2)}) \dots \diff \lambda_{N}(x_{\pi(N)}).
\end{multline*}
Define $\psi(x_{1}, \dots, x_{N})=f(x_{1} * \dots * x_{N})$.  
We start considering the case $\pi(1) < N$. Setting $\pi(1) = j$, for any $(g_{1}, \dots, g_{N})\in G^{N}$  we have 
\begin{align}
&\int \psi((x_{1}, \dots, x_{N})*(g_{1}, \dots, g_{N})) \diff \sigma(x_{1}, \dots, x_{N}) \label{eq:psisigma}\\
&\qquad= \int \psi((x_{1}*g_{1}, \dots, g_{N}*x_{N}))\diff\sigma(x_{1},\dots, x_{N}) \nonumber \\
&\qquad= \int \dots \int \int f(x_{1}*g_{1}* x_{2}*g_{2}\dots *x_{N}*g_{N}) \diff \mu(x_{\pi(1)}) \diff \lambda_{2}(x_{\pi(2)}) \dots \diff \lambda_{N}(x_{\pi(N)}) \nonumber  \\
&\qquad= \int \dots \int \int f(x_{1}*x_{2}*\dots *x_{N}) \diff \mu(x_{\pi(1)}) \diff \lambda_{2}(x_{\pi(2)}) \dots \diff \lambda_{N}(x_{\pi(N)})  \nonumber \\
&\qquad=L, \nonumber 
\end{align}
where the third equality can be shown by setting
\[
\xi(x_{1}, \dots, x_{j-1}, x_{j+1}, \dots, x_{N}) = \int f(x_{1} * \dots * x_{N}) \diff \mu(x_{j})
\]
and noticing that 
\begin{align*}
&\int \dots \int \int f(x_{1}*g_{1} * \dots * x_{j-1}*g_{j-1} * x_{j} * g_{j}* x_{j+1} * g_{j+1} * \dots * x_{N} * g_{N}) \\
& \qquad \diff \mu(x_{j}) \diff \lambda_{2}(x_{\pi(2)}) \dots \diff \lambda_{N}(x_{\pi(N)}) \\
&\qquad = \int \dots \int  \xi(x_{1}*g_{1}, \dots, x_{j-1}*g_{j-1}, g_{j}* x_{j+1} * g_{j+1}, x_{j+2} *g_{j+2},  \dots, x_{N} * g_{N}) \\
& \qquad \qquad  \diff \lambda_{2}(x_{\pi(2)}) \dots \diff \lambda_{N}(x_{\pi(N)}) \\
&\qquad = \int \dots \int  \xi(x_{1}, \dots, x_{j-1}, x_{j+1},  \dots, x_{N} )  \diff \lambda_{2}(x_{\pi(2)}) \dots \diff \lambda_{N}(x_{\pi(N)}) \\
&\qquad= \int \dots \int \int f(x_{1}*x_{2}*\dots *x_{N}) \diff \mu(x_{\pi(1)}) \diff \lambda_{2}(x_{\pi(2)}) \dots \diff \lambda_{N}(x_{\pi(N)}).
\end{align*}
Therefore, we can apply Lemma \ref{lem:minimalgain} to the amenable FC-group $G^{N}$ and to the function $\psi(x_{1}, \dots, x_{N})$. This means that there exists an invariant measure $\rho'\in\mathcal{I}(G^{N})$ such that $\int f(x_{1} * \dots * x_{N}) \diff \rho'(x_{1}, \dots, x_{N}) \ge L>I(f)^{+}$. This contradicts Lemma~\ref{lem:product}.

If $\pi(1) = N$, then the result can be proved replacing \eqref{eq:psisigma} with 
\[
\int \psi((g_{1}, \dots, g_{N}) * (x_{1}, \dots, x_{N})) \diff \sigma(x_{1}, \dots, x_{N}) \\
\]
and following the steps of the previous case.
\end{proof}

\section{Proof of the main result}\label{se:Proof}

As before, $G$ is a countable amenable FC-group and, for $i \in P$, let $\phi_i:G\rightarrow[0,1]$ and $\overline\lambda_{i}\in \mathcal{I}(G)$ be such that
\[
\int \phi_{i} \diff \overline\lambda_{i}=I(\phi_{i})^{+}.
\]

Every abelian group is an amenable FC-group, hence Theorem~\ref{th:main} is a corollary of the following more general result.
\begin{theorem}\label{th:main2}
If $G$ is a countable amenable FC-group, then the profile of strategies $(\overline\lambda_{1}, \dots, \overline\lambda_{N})$ is a Nash equilibrium for the game $\mathcal{G}(P,G, \boldsymbol{\phi}, \nu)$.
\end{theorem}

\begin{proof}
We have to prove that for all $i \in P$ and all $\mu_{i} \in \mathcal{P}(G)$ we have
\begin{equation*}
u_i^{\nu}(\overline{\lambda}_{1}, \dots, \overline{\lambda}_{i}, \dots, \overline{\lambda}_{N}) \ge u_i^{\nu}(\overline{\lambda}_{1}, \dots, \overline{\lambda}_{i-1}, \mu_{i}, \overline{\lambda}_{i+1}, \dots, \overline{\lambda}_{N}).
\end{equation*}
Using \eqref{eq:ualpha}, we know that
\begin{equation*}
u_i^{\nu}(\overline{\lambda}_{1}, \dots, \overline{\lambda}_{i}, \dots, \overline{\lambda}_{N}) = \sum_{\pi \in \Sigma(P)} \nu(\pi) \int_{G} \dots \int_{G} u_{i}(x_{1}, \dots, x_{N}) \diff \overline{\lambda}_{\pi(1)}(x_{\pi(1)})  \dots \diff \overline{\lambda}_{\pi(N)}(x_{\pi(N)}),
\end{equation*}
where $u_{i}(x_{1}, \dots, x_{N}) = \phi_{i}(x_{1} * \dots * x_{N})$.
Since all $\overline{\lambda}_{j}$ are invariant, each of the summands of $u_i^{\nu}(\overline{\lambda}_{1}, \dots, \overline{\lambda}_{i}, \dots, \overline{\lambda}_{N})$ is equal to the corresponding summand for 
\linebreak
$u_i^{\nu}(\overline{\lambda}_{1}, \dots, \overline{\lambda}_{i-1}, \mu_{i}, \overline{\lambda}_{i+1}, \dots, \overline{\lambda}_{N})$, except when $\pi(1)=i$. 
If we call $\Sigma_{i}(P)$  the class of all permutations of $P$ such that $\pi(1)=i$, then all we have to prove is
\begin{multline*}
 \sum_{\pi \in \Sigma_{i}(P)} \nu(\pi) \int_{G} \dots \int_{G }\int_{G} u_{i}(x_{1}, \dots, x_{N}) \diff \overline{\lambda}_{i}(x_{i}) \diff \overline{\lambda}_{\pi(2)}(x_{\pi(2)})  \dots \diff \overline{\lambda}_{\pi(N)}(x_{\pi(N)}) \\
\ge 
\sum_{\pi \in \Sigma_{i}(P)} \nu(\pi) \int_{G} \dots \int_{G} \int_{G} u_{i}(x_{1}, \dots, x_{N}) \diff \mu_{i}(x_{i})  \diff \overline{\lambda}_{\pi(2)}(x_{\pi(2)}) \dots \diff \overline{\lambda}_{\pi(N)}(x_{\pi(N)}).
\end{multline*}
This inequality can be shown to hold summand by summand. More precisely, by Theorem~\ref{th:variationalprinciple} we know that each summand with $\mu_{i}$ is majorized by some invariant measure; but $\bar{\lambda}_{i}$ maximizes the value over all invariant measures, so the result follows.
\end{proof}

\section{Some generalizations}\label{se:Generalization}

Here we consider some generalizations of Theorem~\ref{th:main2}. The first extension allows us to show the existence of an equilibrium in Wald's game, a classical example of countable game that does not admit equilibria in $\sigma$-additive mixed strategies.

\subsection{Transformations of group operations}

Let $\eta_{1}, \dots, \eta_{N}: G \to G$ be bijections and let 
\begin{equation}\label{eq:uphieta}
u^{\boldsymbol{\eta}}_{i}(x_{1}, \dots, x_{N}) = \phi_{i}(\eta_{1}(x_{1}) * \dots * \eta_{N}(x_{N})) \quad \text{for $i \in P$}.
\end{equation}
For $\boldsymbol{\eta} = (\eta_{1}, \dots, \eta_{N})$ call  $\mathcal{G}(P,G, \boldsymbol{\phi}, \boldsymbol{\eta}, \nu)$ the game where the payoffs are given by the mixed extensions of  \eqref{eq:uphieta}, as in \eqref{eq:ualpha}.

\begin{theorem}\label{th:generaleta}
If $G$ is a countable amenable FC-group, then the game $\mathcal{G}(P,G, \boldsymbol{\phi}, \boldsymbol{\eta}, \nu)$ admits a Nash equilibrium.
\end{theorem}

\begin{proof}
Call
\[
u_{i}(y_{1}, \dots, y_{N}) = \phi_{i}(y_{1} * \dots * y_{N}).
\]
By Theorem~\ref{th:main2} we know that the game $\mathcal{G}(P,G, \boldsymbol{\phi}, \nu)$ admits a Nash equilibrium given by $(\overline{\lambda}_{1}, \dots, \overline{\lambda}_{N})$. Therefore if we define for $i\ \in P$ a measure $\rho_{i}$ on $2^{G}$ as follows
\[
\rho_{i}(A) = \overline{\lambda}_{i}(\eta_{i}(A)),
\]
then the profile $(\rho_{1}, \dots, \rho_{N})$ is a Nash equilibrium of  $\mathcal{G}(P,G, \boldsymbol{\phi}, \boldsymbol{\eta}, \nu)$.
\end{proof}

\subsection{Graph games}

Theorem~\ref{th:main} requires that the payoff of each player be a function of the group operation over the strategies of all players. This hypothesis is quite restrictive. The next theorem substantially weakens it by allowing the payoff of player $i$ to be a function of the group operation over a---possibly small---subset of players, provided it includes player $i$ herself and at least another player. This can be interpreted as a game over a graph, where the payoff function of each player depends only on her action and the actions of her neighbors.

For every $i \in P$ let $P_{i} \subset P$ be such that $i \in P_{i}$ and $|P_{i}| \ge 2$.  Consider functions $\phi_{i} : G \to [0,1]$ such that
\begin{equation}\label{eq:uiNi}
u_{i}(x_{1}, \dots, x_{N}) = \phi_{i}(*_{j \in P_{i}} x_{j}),
\end{equation}
that is, the payoff of player $i$ depends only on the strategies of her neighbors.

Call $\boldsymbol{P} = (P_{1}, \dots, P_{N})$ and define the game $\mathcal{G}(P, \boldsymbol{P}, G, \boldsymbol{\phi}, \nu)$ where the payoff functions are as in \eqref{eq:uiNi}.

\begin{theorem}\label{th:generalNi}
If $G$ is a countable amenable FC-group, then the game $\mathcal{G}(P, \boldsymbol{P}, G, \boldsymbol{\phi}, \nu)$ admits a Nash equilibrium.
\end{theorem}

The proof of this theorem follows the line of the proof of Theorem~\ref{th:main2} and is therefore omitted.

\section{Examples}\label{se:Examples}

\subsection{Games on $\mathbb{Z}$}

As the next proposition shows, when the game is played on $\mathbb{Z}$, the equilibrium of Theorem~\ref{th:main2} has an interesting structure, that is,  each equilibrium strategy has its mass either entirely adherent to $-\infty$ or to $+\infty$, or if it is split between the two, then it can be split in any possible way.
\begin{proposition}\label{pr:lambdapmN}
Consider a game  $\mathcal{G}(P, \mathbb{Z}, \boldsymbol{\phi}, \nu)$. 
Let  $(\lambda_{1}, \dots, \lambda_{N})$ be an equilibrium for this game, where  $\lambda_{1}, \dots, \lambda_{N} \in \mathcal{I}(\mathbb{Z})$. Then, for all $i \in P$, one of the following three possibilities is true
\begin{enumerate}[{\rm (a)}]
\item\label{it:pr:lambdapmN-a}
$\lambda_{i}(\mathbb{N})=0$,

\item\label{it:pr:lambdapmN-b}
$\lambda_{i}(\mathbb{N}) = 1$,

\item\label{it:pr:lambdapmN-c}
if $0 < \lambda_{i}(\mathbb{N}) <1$, then for any $\delta \in [0,1]$ there exists another equilibrium strategy $\lambda_{i}' \in \mathcal{I}(\mathbb{Z})$ with $\lambda_{i}'(\mathbb{N}) = \delta$.

\end{enumerate}
\end{proposition}

\begin{proof}

Assume that neither \ref{it:pr:lambdapmN-a} nor \ref{it:pr:lambdapmN-b}  holds. We prove that \ref{it:pr:lambdapmN-c} must be true. Suppose that  $\lambda_{i} \in\mathcal{I}(\mathbb{Z})$ is an equilibrium strategy for player $i$ and $\lambda_{i}(\mathbb{-N}) = \theta \in (0,1)$.
Call
\[
K=\int \phi_{i} \diff \lambda_{i}, \quad \theta K_{1}=\int_{-\mathbb{N}} \phi_{i} \diff \lambda_{i}, \quad (1-\theta) K_{2} = \int_{\mathbb{N}} \phi_{i} \diff \lambda_{i}.
\]
Then 
\[
K = \theta K_{1} + (1-\theta) K_{2}.
\]
Assume that $K_{2} > K_{1}$.
Consider now a measure $\lambda_{i}' \in \mathcal{I}(\mathbb{Z})$ such that $\lambda_{i}'(-\mathbb{N})=0$ and for $A \subset \mathbb{N}$ we have $\lambda_{i}'(A) = (1-\theta)^{-1} \lambda_{i}(A)$. Then 
\[
\int \phi_{i} \diff \lambda_{i}' = K_{2} > K,
\] 
which, by Theorem~\ref{th:main2}, contradicts the fact that $\lambda_{i}$ is an equilibrium strategy. A similar argument holds if $K_{1} > K_{2}$. 

It is easy to see that if $K_{1}=K_{2}$, then any measure $\lambda_{i}''$ such that for $0 < \kappa < \theta^{-1}$
\begin{align*}
\lambda_{i}''(A) &= \kappa \lambda_{i}(A) \quad \text{for $A \subset -\mathbb{N}$}, \\
\lambda_{i}''(A) &= (1-\theta\kappa) (1-\theta)^{-1} \lambda_{i}(A) \quad\text{ for $A \subset \mathbb{N}$},
\end{align*}
satisfies $\int \phi_{i} \diff \lambda_{i}'' = K$ and therefore 
is an equilibrium strategy. This proves part \ref{it:pr:lambdapmN-c}. 
\end{proof}

\begin{example}[Matching pennies]
Consider the following countable version of matching pennies. The strategy set of each of the two players is $\mathbb{Z}$ and the payoff functions are
\[
u_{1}(x,y) = 1-u_{2}(x,y) = \mathds{1}_{2\mathbb{Z}}(x+y),
\]
where $k\mathbb{Z}$ is the set of multiples of $k$.
This game is equivalent to the one where players choose only Odd or Even and player 1 wins if both players make the same choice. Any profile of strategies $(\mu_{1}, \mu_{2})$ such that $\mu_{1}(2\mathbb{Z}) = \mu_{2}(2\mathbb{Z}) = 1/2$ is an equilibrium of the game.

The game can be generalized to $N$ players as follows. Consider a partition $A_{1}, \dots, A_{N}$ of $\mathbb{Z}$ and payoff functions
\[
u_{i}(x_{1}, \dots, x_{N}) = \mathds{1}_{A_{i}}(x_{1}+ \dots + x_{N}).
\]
If for each $i \in P$ the measure 
\[
\overline{\lambda}_{i} \in \arg \max_{\lambda_{i} \in \mathcal{I}(\mathbb{Z})} \lambda_{i}(A_{i}),
\]
then, by Theorem~\ref{th:main2}, the profile $(\overline{\lambda}_{1}, \dots, \overline{\lambda}_{N})$ is a Nash equilibrium of the game.

Notice that, if all sets $A_{1}, \dots, A_{N}$ are periodic (not necessarily with the same period), then for all $\lambda, \lambda' \in \mathcal{I}(\mathbb{Z})$ we have 
\[
\lambda(A_{i}) = \lambda'(A_{i}),
\]
Therefore any profile of invariant measures is an equilibrium. Let $m$ be the lowest common multiple of the periods $m_{i}$ of the $A_{i}$'s. 

Consider now the $m$ congruence classes $m\mathbb{Z} + k$.  
Any profile of probability measures $(\mu_{1}, \dots, \mu_{N})$ such that for $i \in P$ and $k \in \{0, \dots, m-1\}$
\[
\mu_{i}(m\mathbb{Z} + k) = \overline{\lambda}_{j}(m\mathbb{Z} + k) = 1/m
\]
is an equilibrium, too.
\end{example}

\begin{example}[Wald's game]
The following game was  introduced by \citet{Wal:AM1945} as a counterexample to the existence of minmax in zero-sum two-person games when the sets of strategies for both players are infinite.
Let the strategy set be $\mathbb{Z}$ and
\[
u_{1}(x,y) = 1 - u_{2}(x,y) =
\begin{cases}
1 & \text{if $x > y$}, \\
1/2 & \text{if $x = y$}, \\
0 & \text{if $x < y$}.
\end{cases}
\]
Call $z := -y$; then the payoff function becomes
\[
u_{1}(x,z) =
\begin{cases}
1 & \text{if $x+z > 0$}, \\
1/2 & \text{if $x+z = 0$}, \\
0 & \text{if $x+z < 0$}.
\end{cases}
\]
Then $u_{1}(x,z)$ is an example of $\phi(x+z)$ with $+$ as the group operation.
Applying Theorem~\ref{th:generaleta} we obtain the equilibrium $(\lambda, \rho)$ with $\lambda, \rho \in \mathcal{I}(\mathbb{Z})$ and $\lambda(\mathbb{N}) = \rho
(-\mathbb{N}) = 1$. This shows the striking difference between countably additive and finitely additive extensions of countable games.
\end{example}

\subsection{Games on $\mathbb{Z}^{2}$}

Consider a game where the strategy set of each player is $\mathbb{Z}^{2}$ and for $i \in \{1, \dots, N\}$ the payoff function $u_{i}$ is
\begin{equation}\label{eq:uiZ2}
u_{i}(x_{1} + \dots + x_{N}) = \mathds{1}_{C_{i}}(x_{1} + \dots + x_{N}),
\end{equation}
where $C_{i}$ is some open cone, i.e., $C_{i}$ is an open subset of $\mathbb{R}^{2}$ such that if $x \in C_{i}$, then $\beta x \in C_{i}$ for all $\beta \in \mathbb{R}_{+}$.

\begin{proposition}
Let $\phi_{i} = \mathds{1}_{C_{i}}$. If  $\overline{\lambda}_{i}\in \mathcal{I}(\mathbb{Z}^{2})$ and $\overline{\lambda}_{i}(C_{i})=1$ for all $i \in P$, then  the profile $(\overline{\lambda}_{1}, \dots, \overline{\lambda}_{N})$ is an equilibrium of the game $\mathcal{G}(P, \mathbb{Z}^{2}, \boldsymbol{\phi}, \nu)$.\end{proposition}

\begin{proof}
By Theorem~\ref{th:main2} a measure $\overline{\lambda}_{i}$ is an equilibrium strategy if 
\[
\overline{\lambda}_{i} = \arg \max_{\lambda_{i} \in \mathcal{I}(\mathbb{Z}^{2})} \lambda_{i}(C_{i}).
\]
Since $\lambda_{i}(C_{i}) \le 1$ all we need to prove is that for each open cone $C_{i}$ there exists an invariant measure $\overline{\lambda}_{i}$ such that  $\overline{\lambda}_{i}(C_{i}) = 1$. This is achieved by using F{\o}lner sequences as follows.
For every open cone $C_{i}$ there exists a convex open cone $C_{i}^{*} \subset C_{i}$. We now prove the existence of  $\overline{\lambda}_{i}$ such that $\overline{\lambda}_{i}(C_{i}^{*}) = 1$. Call $Q_{n} := \{-n, \dots, n\}^{2} \subset \mathbb{Z}^{2}$ and $F_{n} := Q_{n} \cap C_{i}^{*}$. 
\begin{claim}\label{cl:Folner}
The sequence $F_{n}$ is a F{\o}lner sequence.
\end{claim}
\begin{proof}[Proof of Claim~\ref{cl:Folner}]
We need to prove that for all $g \in \mathbb{Z}^{2}$ we have
\[
\lim_{n \to \infty} \frac{|(g*F_{n}) \triangle F_{n}|}{|F_{n}|} = 0.
\]
There exists $\delta > 0$ such that for $n$ large enough $|F_{n}| > \delta n^{2}$. Moreover if $g=(g_{1}, g_{2})$, then $|(g*F_{n}) \triangle F_{n}| \le 4|g_{1} g_{2}| n$. This proves the Claim.
\end{proof}
Define a probability measure $\mu_{n}$ on $2^{\mathbb{Z}^{2}}$ as follows:
\[
\mu_{n}(A) =  \frac{|A \cap F_{n}|}{|F_{n}|}.
\]
Call $\overline{\lambda_{i}}$ a weak* limit of a subnet $\mu_{c(\alpha)}$ of $\mu_{n}$. First we prove that  $\overline{\lambda_{i}} \in \mathcal{I}(\mathbb{Z}^{2})$. For all $A \subset \mathbb{Z}^{2}$ and for all $g \in \mathbb{Z}^{2}$ 
\begin{align*}
|\overline{\lambda_{i}}(A+g) - \overline{\lambda_{i}}(A)| &= \lim_{\alpha} \left| \frac{|(A+g) \cap F_{c(\alpha)}|}{|F_{c(\alpha)}|} -  \frac{|A \cap F_{c(\alpha)}|}{|F_{c(\alpha)}|} 
\right| \\
&=  \lim_{\alpha} \left| \frac{|A \cap (F_{c(\alpha)}-g)|}{|F_{c(\alpha)}|} -  \frac{|A \cap F_{c(\alpha)}|}{|F_{c(\alpha)}|} 
\right| \\
& \le \lim_{\alpha}\left| \frac{|A \cap ((F_{c(\alpha)}-g) \triangle F_{c(\alpha)}) |}{|F_{c(\alpha)}|} 
\right|  \\
&= 0,
\end{align*}
where the inequality is due to the fact that 
\[
A \cap ((F_{c(\alpha)}-g) \setminus  (A \cap F_{c(\alpha)})) \subset A \cap ((F_{c(\alpha)}-g) \triangle F_{c(\alpha)}).
\]
Then we prove that $\overline{\lambda}_{i}(C_{i}^{*}) = 1$. Indeed
\[
\overline{\lambda}_{i}(C_{i}^{*}) = \lim_{\alpha} \frac{|C_{i}^{*} \cap F_{c(\alpha)}|}{|F_{c(\alpha)}|}.
\]
Since $F_{c(\alpha)} \subset C_{i}^{*}$ for all $\alpha$, we have 
\[
 \frac{|C_{i}^{*} \cap F_{c(\alpha)}|}{|F_{c(\alpha)}|} = 1 \quad \text{for all $\alpha$}. \qedhere
\]
\end{proof}

\subsection{Games on $\mathbb{Q} \cap [0,1] \mod 1$}\label{suse:Q01}

Take $G = \mathbb{Q} \cap [0,1]$ equipped with the sum modulo $1$. Then $G$ is a countable abelian group, so it is an amenable FC-group. Observe that any invariant mean $\lambda$ on $G$ satisfies the following property. 
With an abuse of language we use the  symbol $[a,b]$ to denote the set $\{x \in G: a \le x \le b\}$.

\begin{proposition}\label{pr:lambda01}
For every $\lambda \in \mathcal{I}(G)$ and every $a, b \in [0,1]$, $a < b$, we have $\lambda([a,b]) = b-a$.
\end{proposition}

\begin{proof}
Observe that, by invariance, for every $n \in \mathbb{N}$
\[
\lambda\left(\left[0,\frac{1}{n}\right]\right) = \lambda\left(\left[\frac{1}{n}, \frac{2}{n}\right]\right) =  \dots =\lambda\left(\left[\frac{n-1}{n}, 1\right]\right) = \frac{1}{n}.
\]
Hence, by finite additivity, if $a = k/n$ and $b = m/n$, then
\[
\lambda([a,b]) = \lambda\left(\left[\frac{k}{n}, \frac{m}{n} \right] \right) = \frac{m-k}{n}.
\]
\end{proof}
Therefore any two invariant means on $G$ coincide on the algebra generated by intervals in $G$, but they can be extended in many different ways to $2^{G}$.

\begin{proposition}
Call $\mathfrak{c}$ the cardinality of $\mathbb{R}$. Then $|\mathcal{I}(G)| \ge 2^{\mathfrak{c}}$.
\end{proposition}

\begin{proof}
If $G$ is endowed with the discrete topology, then it is a locally compact group, whose Haar measure is the counting measure. It follows that the set of essentially (with respect to the Haar measure) bounded real-valued functions on $G$ is equal to $\ell^{\infty}(G)$. Therefore \citet[Theorem on page 444]{Cho:TAMS1970} can be used to get the result.
\end{proof}

Consider now a game on $G$ where for $i \in P$,
\[
\phi_{i}(x) = \mathds{1}_{A_{i}}(x).
\]
If $A_{1}, \dots, A_{N}$ are in the algebra generated by intervals, then every profile $(\lambda_{1}, \dots, \lambda_{N})$, with $\lambda_{i} \in \mathcal{I}(G)$ is an equilibrium. Otherwise the equilibrium strategy for player $i$ is 
\[
\overline{\lambda}_{i} \in \arg\max_{\lambda \in \mathcal{I}(G)} \lambda(A_{i}).
\]

\begin{example}[Rock-scissors-paper]
The classical rock-scissors-paper game is played by two players whose strategy set is $G=\{R, S, P\}$. The payoff for player $1$ is
\begin{align}\label{eq:uRSP}
u_{1}(R,S)=u_{1}(P,R)=u_{1}(S,P)&=1, \nonumber\\
u_{1}(R,R)=u_{1}(P,P)=u_{1}(S,S)&=1/2, \\
u_{1}(S,R)=u_{1}(R,P)=u_{1}(P,S)&=0,  \nonumber
\end{align}
and $u_{2}(x,y) = 1- u_{1}(x,y)$.
It is well known that the unique equilibrium of this game is the profile of uniform mixed strategies for each player.
  
We can endow the set  $G$ with an operation $*$ that makes it an abelian group, as follows:
\begin{align*}
R*R = P*S = S*P &= R \\
R*P = P*R = S*S &= P \\
R*S = S*R = P*P &= S.
\end{align*}
Note that $R$ is the unit element and 
\[
R^{-1}=R, P^{-1}=S, S^{-1}=P.
\]
Actually the group $G$ can be identified with $\mathbb{Z}/3 \mathbb{Z} = \{\bar{0}, \bar{1}, \bar{2}\}$ using the following isomorphism $\Phi$: 
\[
\Phi(R)=\bar{0},\quad \Phi(P)=\bar{1},\quad \Phi(S)=\bar{2}.
\]

Therefore \eqref{eq:uRSP} can be re-written as follows. 
\[
u_{1}(x,y) = 
\begin{cases}
1 & \text{if $y*x^{-1} = S$,}\\
1/2  & \text{if $y*x^{-1} = R$,}\\
0 & \text{if $y*x^{-1} = P$.}
\end{cases}
\]
Thus rock-scissors-paper is a game with payoffs of the form \eqref{eq:uphieta}.

The game can be generalized to a countable setting as follows.
Let the strategy set of each of the two players be the group $G$ equal to $[0,1] \cap \mathbb{Q}$ equipped with the sum $\mod 1$ and let, for $0 < \alpha < \beta < 1$, the payoff function be
\[
u_{1}(x,y) = 
\begin{cases}
1 & \text{if $\beta < y-x < 1$,}\\
1/2  & \text{if $\alpha \le y-x \le \beta$, or $y-x=0$,}\\
0 & \text{if $0 < y-x < \alpha$.}
\end{cases}
\]

Combining Theorem~\ref{th:generaleta} and Proposition~\ref{pr:lambda01} we can show that every pair of invariant means is an equilibrium. 
To wit, the proof of  Theorem~\ref{th:generaleta} shows that $(\bar{\lambda}_{1}, \bar{\lambda}_{2})$ is an equilibrium if 
\begin{enumerate}[(a)]
\item\label{it:RSP-a}
$\bar{\lambda}_{2} \in \mathcal{I}(G)$,

\item\label{it:RSP-b}
$\widetilde{\lambda}_{1} \in \mathcal{I}(G)$, where $\widetilde{\lambda}_{1}(A)=\bar{\lambda}_{1}(-A)$,

\item\label{it:RSP-c}
$\widetilde{\lambda}_{1} \in \arg \max_{\lambda \in \mathcal{I}(G)} \int \phi_{1} \diff \lambda$, where
\[
\phi_{1}(x) = \mathds{1}_{(\beta, 1)}(x) + \frac{1}{2} \mathds{1}_{[\alpha,\beta]}(x) + \frac{1}{2} \mathds{1}_{\{0\}}(x),
\]

\item\label{it:RSP-d}
$\bar{\lambda}_{2} \in \arg \max_{\lambda \in \mathcal{I}(G)} \int (1-\phi_{1}) \diff \lambda$. 
\end{enumerate}

Proposition~\ref{pr:lambda01} shows that every invariant mean satisfies \ref{it:RSP-c} and \ref{it:RSP-d}. Furthermore, since $\widetilde{\lambda}_{1} \in \mathcal{I}(G)$ we have,  $\bar{\lambda}_{1} \in \mathcal{I}(G)$, too. Therefore every pair of invariant means satisfies \ref{it:RSP-a}--\ref{it:RSP-d} and hence is an equilibrium.
\end{example}

\begin{example}[Love and hate]
This game is played by an even number of players $N=2k$. The strategy set of each player is $\mathbb{Q} \cap [0,1] \mod 1$. The payoff functions have this form for $h \in \{1, \dots, k\}$
\begin{align*}
u_{2h}(x_{1}, \dots, x_{2k}) = -d(x_{2h}, x_{2h+1}), \\
u_{2h+1}(x_{1}, \dots, x_{2k}) = d(x_{2h+1}, x_{2h+2}), \\
\end{align*}
where $N+j := j$ and
\[
d(x,y) = \min(|x-y|,1-|x-y|).
\]
In words, every even player wants to be as close as possible to the following odd player and every odd player wants to be as far as possible from the following even player.

If we define 
\[
\eta_{2h}(x) = x, \quad
\eta_{2h+1}(x) = -x,
\]
then the payoffs can be written as
\begin{align*}
u_{2h}(x_{1}, \dots, x_{2k}) &= \phi_{2h}(\eta_{2h}(x_{2h})+\eta_{2h+1}(x_{2h+1})),\\
u_{2h+1}(x_{1}, \dots, x_{2k}) &= \phi_{2h+1}(\eta_{2h+1}(x_{2h+1})+\eta_{2h+2}(x_{2h+2})).
\end{align*}
Combining Theorems~\ref{th:generaleta} and \ref{th:generalNi} we obtain that there exist equilibria that are invariant means for each player. 
\end{example}

\begin{remark}
Theorems~\ref{th:generaleta} and \ref{th:generalNi} and especially their combination broaden considerably the class of games for which existence of equilibria can be shown using group-theoretic arguments.

On one hand some of the group games that we have considered are the countable version of finite games that have an equilibrium in uniform strategies, see, e.g., matching pennies and rock-scissors-paper. On the other hand the finite version of some other countable group games does not have an equilibrium in uniform strategies, see, e.g., Wald's game. This shows that the class of countable group games is not a trivial extension of the class of finite group games.

\end{remark}

\section{Uncountable games}\label{se:Uncountable}

In this section we consider games $\mathcal{G}(P,G,u,\nu)$, where $G$ is an uncountable group. Some hypotheses on $G$ will be needed to generalize the results of the previous sections. In particular we will require $G$ to be a locally compact metric group, which means that the group is equipped with a locally compact metrizable topology that is compatible with the group operation.\footnote{As examples of such a group consider for instance $(\mathbb{R}, +)$ and $S^{1}$, the unit circle in $\mathbb{R}^{2}$, where the group operation is the angle sum $\mod 2 \pi$.} In this case, a classical theorem by \citet{Haa:AM1933} guarantees the existence of a unique invariant countably additive measure on $G$. This measure is finite if and only if the group $G$ is compact, therefore in general this is not a mixed strategy. Nevertheless the Haar measure was used by \citet{vNe:FM1929} to define amenability.

Denote by $L^{\infty}(G)$ the Banach space of all real-valued functions on $G$ that are essentially bounded with respect to the Haar measure. Given $g\in G$ and $f\in L^{\infty}(G)$, we denote by $L_{g}f$ and $R_{g}f$ respectively the left- and the right-translation of $f$ by $g$, i.e.,
\[
(L_{g}f)(x)=f(g*x)\quad\text{and}\quad (R_{g}f)(x)=f(x*g).
\]

\begin{definition}\label{defin:amenableuncountable}
A locally compact topological group $G$ is called amenable if there exists a linear operator $T:L^{\infty}(G)\to\mathbb R$ verifying the following properties:
\begin{itemize}
\item[] \textbf{Positivity.} If $f:G\to\mathbb R_{+}$, then $T(f)\geq0$.
\item[] \textbf{Normalization.} If $f \equiv 1$, then $T(f)=1$.
\item[] \textbf{Invariance.} For all $f\in L^\infty(G)$ and for all $g\in G$, one has 
\begin{equation}\label{eq:invariantoperator}
T(L_{g}f)=T(f)=T(R_{g}f).
\end{equation}
\end{itemize}
A positive and normalized linear operator $T$ that verifies the first (second) equality in \eqref{eq:invariantoperator} is called \emph{left-invariant  (right-invariant) mean}; an operator that is both left- and right-invariant is called \emph{invariant mean}.
\end{definition}

Any positive and normalized linear operator $T:L^{\infty}(G)\to\mathbb R$ defines a finitely additive probability measure $\mu$ on $\mathcal{B}$, the $\sigma$-algebra of Borel subsets of $G$, as follows
\[
\mu(A) = T(\mathds{1}_{A}), \quad\text{for all $A \in \mathcal{B}$}.
\]
Therefore, a locally compact group is amenable if and only if there exists a finitely additive probability measure on $\mathcal{B}$ which is invariant with respect to the group operation. 

Now we consider a game  $\mathcal{G}(P,G,u,\nu)$, where for every $i \in P$
\[
u_{i}(x_{1}, \dots, x_{N}) = \phi_{i}(x_{1}* \dots * x_{N}),
\]
for some Borel-measurable function $\phi_{i}$ which is assumed to be bounded and integrable with respect to every finitely additive probability measure on $\mathcal{B}$.
The mixed extension of the game is achieved using \eqref{eq:ualpha} as in the countable case.

In order to prove the existence of equilibria we need the following two results. The first is a more general version of F{\o}lner's theorem.

\begin{theorem}[\citet{Fol:MS1955}]\label{th:generalforlner}
A locally compact group $G$ with Haar measure $\mu$ is amenable if and only if there is a sequence of compact subsets $F_n$ of $G$ such that $\mu(F_n)\to\infty$ and 
\[
\frac{\mu(F_{n} * (g\triangle F_{n}))}{\mu(F_{n})}\to 0 \quad \text{for all $g\in G$}.
\]
\end{theorem}

The following theorem is folklore, \citep[see, e.g.,][Theorem 4.3]{Mer:ZRB2000}.

\begin{theorem}\label{th:weakstar}
The set of finitely supported probability measures on a metric space is dense, with respect to the weak* topology, in the set of all finitely additive probability measures.
\end{theorem}

We can now state our general existence result for group games.

\begin{theorem}\label{th:maingeneral}
Let $G$ be an amenable, locally compact, metric group such that left-invariant and right-invariant means coincide. Then the game $\mathcal {G}(P,G,u,\nu)$ admits Nash equilibria which do not depend on $\nu$. 
\end{theorem}

The proof of Theorem~\ref{th:maingeneral} follows the line of the proof of Theorem~\ref{th:main2}, using Theorems~\ref{th:generalforlner} and \ref{th:weakstar} and is therefore omitted. 

The following general version of Paterson's theorem guarantees that the class of groups that satisfy the hypotheses of Theorem~\ref{th:maingeneral} is large.

\begin{theorem}[\citet{Pat:PJM1979}]
An amenable, compactly generated\footnote{A topological group is called \emph{compactly generated} if it is generated by a compact subset; namely, there is a compact subset $K$ of $G$ such that every element $g$ of $G$ can be written in the form $g=k_1*\ldots*k_n$, with $k_i\in K\cup K^{-1}$. For instance, $(\mathbb R,+)$ is compactly generated by the interval $[-1,1]$: every real number $r$ can be written in the form $r= s+\ldots+s$, where $s\in[-1,1]$.}, locally compact group $G$ has the property that right-invariant and left-invariant means coincide if and only if the closure of each conjugacy class is compact.
\end{theorem}

Wald's game can also be played on $\mathbb{R}$ and all games in Subsection~\ref{suse:Q01} can also be played on $[0,1] \mod 1$. 
Theorem~\ref{th:maingeneral} guarantees existence of equilibria in this uncountable setting. All results go through with the suitable needed modifications.

%
%
%
%

\section{Conclusions}\label{se:Conclusions}

We have considered a class of games where the strategy set of each player is a group and the payoff functions depend on the strategies only through the group operation. We have shown that finitely additive equilibria exist for this class of games. In the case of countable groups we have not used any topological conditions, just the algebraic structure of the payoffs. The only measure theoretical assumption refers to the selection of the product of finitely additive mixed strategies. 

Even if the algebraic condition that we use leaves out a huge set of games, it includes cases that are not covered by 
\citet{Mar:IJGT1997}. In fact in general the payoff functions that we considered are not measurable with respect to the algebra generated by the cylinders. Indeed, his payoff  functions satisfy Fubini's theorem \citep[Proposition~3]{Mar:IJGT1997}, whereas ours in general do not. For instance take the two-person zero-sum game on $(\mathbb{Z}, +)$ where $u_{1}(x,y)=\mathds{1}_{\mathbb{N}}(x+y)$. This function is clearly not measurable with respect to the algebra generated by the cylinders, although it is measurable with respect to the $\sigma$-algebra generated by the cylinders. Measurability assumptions on the payoff functions are crucial to prove Marinacci's theorems, whereas we based our proofs on the algebraic properties of the payoffs.

As \citep[Example~2.1]{Sti:GEB2005} shows, the games that we considered in general are not nearly compact and continuous as the ones in \citet{HarStiZam:GEB2005}.

\citet{Sti:GEB2005}
proves very general deep existence results, that are typically non-constructive. In our paper we characterize the equilibrium strategies in a simple form. 

\citet{CapMor:IJGT2012} prove that invariant means are minmax strategies for zero-sum two-person games when the set of allowed strategies is restricted so that the exchange of the order of integration is possible. In our paper we do not put any restriction on the set of mixed strategies. Our results are not only more general, they also requires different tools for their proof.

\citet{CanMenOzdPar:MOR2011} show that any finite game can be decomposed into three components, a potential, a nonstrategic, and a harmonic component. The last component has the property that a profile of uniform strategies is always an equilibrium. We notice that, even if group games share the same property, they are neither a subclass nor a superclass of the class of harmonic games. For instance, if $\phi_{1}= \dots = \phi_{N}$ the group game is a potential game and therefore cannot be harmonic. On the other hand in a harmonic game the strategy sets for the different players are not necessarily the same.

We note that the class of group games is a subspace in the class of all games and that the \citet{CanMenOzdPar:MOR2011} decomposition holds in this subspace by projection.

\subsection*{Acknowledgments} 
The authors thank Marco Dall'Aglio for sparking their interest in group games, Patrizia Berti and Pietro Rigo for their useful comments about finitely additive probability measures, Alain Valette for helpful discussions about amenability, Kent Morrison for pointing out a mistake in a proof, and Ozan Candogan for his insights on harmonic games.

\bibliographystyle{artbibst}
\bibliography{bibsumgame}

\end{document}